\documentclass[twocolumn]{autart}    

\newcommand{\R}{ \mathbb{R}}
\newcommand{\dn}{ \mathbf{d}}
\newcommand{\zero}{\mathbf{0}}
\usepackage{color}
\usepackage{amssymb,amsmath}
\usepackage{theoremref}

\usepackage{comment  }
\usepackage{multirow}
\newenvironment{proof}{\textbf{Proof}:}{\hfill$\square$}
\newtheorem{theorem}{Theorem}
\newtheorem{definition}{Definition}

\newtheorem{remark}{Remark}

\newtheorem{corollary}{Corollary}

\usepackage{graphicx}
\usepackage{xcolor} 

\newcommand{\Yu}{\color{orange}}
\newcommand{\eg}{\textit{e.g.}}
\newcommand{\ie}{\textit{i.e.}}
\usepackage[colorlinks,
linkcolor=blue,
anchorcolor=blue,
citecolor=blue
]{hyperref}
\usepackage[round]{natbib}

\begin{document}

\begin{frontmatter}

\title{Generalized Homogeneous Rigid-Body Attitude Control} 

\thanks[footnoteinfo]{This paper was presented at European Control Conference 2022. Corresponding author Yu Zhou.}

\author[Inria]{Yu Zhou}\ead{yu.zhou@inria.fr},    
\author[Inria]{Andrey Polyakov}\ead{andrey.polyakov@inria.fr},               
\author[Inria]{Gang Zheng}\ead{gang.zheng@inria.fr}  

\address[Inria]{Inria, Univ. Lille, CNRS, Centale Lille, France}  


\begin{abstract}                          
The attitude tracking problem for a full-actuated rigid body in 3D is studied using a 
system model based on Lie algebra $\mathfrak{so}(3)$. 
A nonlinear homogeneous  controller is designed to track globally a smooth attitude trajectory in a finite or a (nearly) fixed time. A global settling time estimate is obtained,  which is easily adjustable by tuning the homogeneity degree. The  input-to-state stability of the control system with respect to measurement noises and additive perturbations is studied. Simulations illustrating the performance of the proposed algorithm are presented.
\end{abstract}

\end{frontmatter}

\section{Introduction}
The problem of a rigid body attitude control is well-known in a variety of applications such as satellites, robots, and aerial/underwater vehicles.
There are three commonly-used global mathematical descriptions  of the rigid body's attitude (orientation): Special Orthogonal  Group  $SO(3)$, quaternion,   Lie algebra $\mathfrak{so}(3)$ (also known as exponential coordinates or  Chevalley's coordinates, see  \cite{chaturvedi2011rigid}).
Indeed, the manifold $SO(3)$ consists of all possible rotation matrices in 3D space. 
 The quaternion also represent globally an orientation of a rigid body but with an ambiguity, which may cause an unwinding phenomenon in attitude control \cite{wen1991attitude}. 
Lie algebra $\mathfrak{so}(3)$ is a minimal (three parameters) representation which contains ambiguities, as it was stated in \cite{stuelpnagel1964parametrization} that any representation with three parameters cannot uniquely represent the global orientation of a rigid body.

Based on the above mentioned three representations, different attitude control techniques have been developed in the literature.
Notice that  a topological structure of the group precludes the existence of continuous time-invariant feedback control for a global asymptotic attitude stabilization \citet{bullo1995proportional,chaturvedi2011rigid,mayhew2013synergistic}.  Several switching-based  algorithms of a global asymptotic/exponential attitude control design have been developed for the system described by quaternions (see, \eg,  \cite{mayhew2011quaternion}, \cite{su2011globally}); 
 $SO(3)$ group (see, \eg, \citet{mayhew2013synergistic}, \cite{lee2015global},  \cite{berkane2017hybrid}) and  
 the Lie algebra $\mathfrak{so}(3)$ (see, \eg, \citet{bharadwaj1998geometry} and \citet{yu2016global}). 

A non-asymptotic (\text{i.e.} finite/fixed-time) attitude control of rigid body has also been intensively studied in the last decade.
 However, due to the topological constraint, non-asymptotic  stabilization and robustness analysis on the manifold $SO(3)$ are not trivial. Thus, most of non-asymptotic attitude control algorithms are currently based on quaternion and Lie algebra $\mathfrak{so}(3)$. 
  There are two approaches to non-asymptotic attitude control design: direct Lyapunov methods (see, \eg, \cite{du2011finite},  \cite{zou2013finite}, \cite{Bing_Acta_2018}, \cite{Chen_TAES_2018},  \cite{zou2020fixed}) and homogeneity-based methods.
 The direct Lyapunov methods provide an estimate of the settling time but it involves a complicated Lyapunov construction. Homogeneity-based methods simplify the design of finite-time control because an asymptotically stable homogeneous system with a negative degree is finite-time stable.
  Several attitude controllers inspired by the ideas of  homogeneity can be found in the literature.  They provide local 
(\cite{du2011finite}, \cite{zou2013finite}, \cite{guo2019robust}, \cite{gui2015simple}) and 
a global (\cite{du2012finite}, \cite{gui2016global}, \cite{shi2018global}) stabilization, and combined with 
sliding mode algorithms in  (\cite{zhu2011attitude}, \cite{tiwari2015rigid}, \citet{shi2017almost}, \cite{guo2019robust}).
	
	The mentioned homogeneity-based finite-time global attitude control design follows the idea of combining a negative-degree locally homogeneous feedback with globally asymptotically stabilizing algorithm \cite{hong2001output}. The main drawbacks of this approach are absence of the settling time estimate; sophisticated  Lyapunov analysis; and impossibility of the fixed-time stabilzation. 
	The aforementioned observation motivates us to develop a new homogeneity-based attitude control.

 Continuous homogeneous control systems   in $\R^n$ have the following advantages:
 finite-time and fixed-time convergence rates,  \cite{BhatBernstein2005:MCSS}, \cite{Andrieu_etal2008:SIAM_JCO};
	 local stability is equivalent to  global one even in a nonlinear case, \cite{Zubov1958:IVM}, \cite{Rosier1992:SCL};
	 homogeneity ensures robustness with respect to a larger class of uncertainties \cite{Hong2001:Aut}, \cite{Andrieu_etal2008:SIAM_JCO};
	 an elimination of an unbounded peaking effect, \cite[Chapter 1]{Polyakov2020:Book} and a reduction of overshoots (\cite{PolyakovKrstic2023:TAC}).

	The aim of this paper is to develop a homogeneous attitude control on Lie algebra. However, the global model on Lie algebra $\mathfrak{so}(3)$ is a specific \textit{impulsive} system ( \citet{bharadwaj1998geometry}). The jump behavior may destroy the above mentioned advantages. There are few studies on impulsive homogeneous system and its stability (\eg, \cite{TunaTeel2006:CDC}, \cite{Goebel_CDC2008},  \cite{forni2010stability}), but they require the scaling invariance of jump set (\ie, if the state in the jump set, then arbitrarily scaled state must remain in the jump set). This condition is not satisfied in the case of attitude control.
	This motivates us to study a specific class of impulsive system that admit an analysis based on the homogeneous implicit Lyapunov function method \cite{Polyakov_etal2015:Aut}.  To the best of authors' knowledge, no research has been done on this class of systems before.

For homogeneous control design, the implicit Lyapunov-function-based methods \cite{AdamyFlemming2004:Aut}, \cite{Korobov1979:DAN} are well-developed for  robust homogeneous  finite/fixed-time stabilization of  continuous-time linear control systems  (\cite{Polyakov_etal2015:Aut}, \cite{Polyakov2020:Book}, \cite{Zimenko_etal2020:TAC}), 
this paper extends the mentioned control technique to  a nonlinear  ( impulsive) system and applies it to an attitude control design.  
The main contributions of our work are as follows.
\begin{itemize}
	\item  The \textit{homogeneous} implicit Lyapunov function method is extended to a specific class of  impulsive systems such that 
		various convergence rates (finite-time, exponential or  nearly fixed-time) can be assessed by a proper selection of the  homogeneity degree. An estimate of a convergence time is given.

	\item A homogeneous attitude control for global attitude tracking  is developed using the implicit Lyapunov function technique. The design procedure  repeats the linear control case but allows the tracking error to converge to zero in a finite/fixed time  or exponentially dependently of the homogeneity degree. The convergence time is estimated.  The homogeneous algorithm becomes the classical linear PD controller (see \eg, \cite{bullo1995proportional}, \cite{yu2016global}) on exponential coordinates in a particular case.
	\item  Globally fixed-time stable attitude controller is designed  by a commutation of two homogeneous controllers as in \cite{Polyakov_etal2015:Aut}. Robustness ( Input-to-State Stability)  of the  closed-loop  system with respect to various perturbations is studied.
\end{itemize}



A preliminary version of this work has been presented at European Control Conference,  \cite{YuECC2022},  where local homogeneous attitude control has been developed. The present paper extends this local result to a global homogeneous attitude controller. The paper is organized as follows. The problem of attitude tracking is stated in Section \ref{sec:prob}.  In Section \ref{sec:homo}, some preliminaries about generalized  homogeneous stabilization and new results about homogeneous stabilization of linear  impulsive systems are given. 
A homogeneous attitude tracking control is design in Section \ref{sec:resu}. The performance and robustness of the developed control algorithm are demonstrated in Section \ref{sec:simu} on numerical simulations. Some properties of $SO(3)$ group and $\mathfrak{so}(3)$ algebra are given in appendix.
\section*{Notation}\vspace{-2mm}
$\mathbb{R}$ is the field of reals, $\mathbb{R}_{+}=\{x \in \mathbb{R}: x \geq 0\}$;    $|\cdot|$ is the  Euclidean  norm in $\mathbb{R}^{n}$ ; $\mathbf{0}$ denotes the zero element of a vector space; $\operatorname{diag}\left\{\lambda_{i}\right\}_{i=1}^{n}$ is the diagonal matrix with elements $\lambda_{i} ; P \succ 0(\prec 0, \succeq 0, \preceq 0)$ for $P \in \mathbb{R}^{n \times n}$ means that the matrix $P$ is symmetric and positive (negative) definite (semi-definite); $C(X, Y)$ denotes the space of continuous functions $X \mapsto Y$, where $X, Y$ are subsets of normed vector spaces; $C^{p}(X, Y)$ is the space of functions continuously differentiable at least up to the order $p$; $\lambda_{\min }(P)$ and $\lambda_{\max }(P)$ represent the minimal and maximal eigenvalue of a matrix $P$; for $P \succeq 0$ the square root of $P$ is a matrix $M=P^{\frac{1}{2}}$ such that $M^{2}=P$; 
$SO(3)\subset \R^{3\times 3}$ is a special orthogonal group (see Appendix)
and $\mathfrak{so}(3)$ is the corresponding Lie algebra consisting of the skew-symmetric $3\times 3$ matrices;
	the mapping  $(\cdot)^{\wedge}: \mathbb{R}^3 \mapsto \mathfrak{so}(3)$ is given by
	$\scriptsize x^{\wedge}\!=\left[\begin{smallmatrix}
		0 & -x_3 & x_2\\
		x_3 & 0 & -x_1\\
		-x_2& x_1& 0
	\end{smallmatrix}\right],$  where $x\!=\![x_1 , x_2 , x_3]^{\top}\in \R^3,$
	and $(\cdot)^{\vee}: \mathfrak{so}(3) \mapsto \mathbb{R}^3$ is the inverse mapping to $(\cdot)^{\wedge}$; $\mathrm{Exp}(\cdot)=\mathrm{exp}((\cdot)^{\wedge})$ defines a mapping from $\mathbb{R}^3$ to $SO(3)$, $\mathrm{Log}(\cdot)=(\mathrm{log}(\cdot))^{\vee}$ defines a mapping from $SO(3)$ to $\mathbb{R}^3$, where $\mathrm{exp}$ and $\mathrm{log}$ are the matrix exponent and the matrix logarithm, respectively;
	$\mathcal{K}$ denotes a class of continuous strictly increasing positive definite functions $\mathbb{R}_+\mapsto\mathbb{R}_+$ with $\alpha(0)=0$. $\mathcal{KL}$ denotes a continuous function $\beta:\mathbb{R}_+\times\mathbb{R}_+\mapsto\mathbb{R}_+$, if, for each fixed s, $\beta(\cdot,s)\in\mathcal{K}$ and for each fixed $r$, $\beta(r,\cdot)$ is strictly decreasing to zero.  
	$\mathbf{B}^n[r]=\left\{x\in\mathbb{R}^n:|x|\le r\right\}$ and $\mathbf{B}^n(r)=\left\{x\in\mathbb{R}^n:|x|< r\right\}$ denote the closed ball and open ball of the radius $r$ on $\mathbb{R}^n$.  $\partial \mathcal{C}$ and $\overline{\mathcal{C}}$ denote the boundary and closure of an open set $\mathcal{C}\in \R^n$. $L^\infty$ denotes the function space that its elements are the essentially bounded measurable function. For $f\in L^{\infty}$, $\|f\|_{\infty}=\inf \left\{c \in \mathbb{R}_{\geq 0}:|f(x)| \leq c \text { for almost every } x\right\}$.

\section{Problem statement}
\label{sec:prob}
The dynamics of a  fully-actuated rigid body rotation motion is governed by the nonlinear differential equation (see, e.g. \citet{bharadwaj1998geometry}):
\begin{subequations}\label{eq:dyn}
	\begin{align}
		\dot{R} &= R\omega^{\wedge}, \quad t>0, \quad R(0)=R_0, \label{eq:dyn1}\\
		\dot{\omega}&=J^{-1}\left( -\omega\times J\omega + M\right), \quad \omega(0)=\omega_0,
		\label{eq:dyn2}
	\end{align}
\end{subequations}
where $R(t)\in SO(3)$ is the rotation matrix with respect to the inertial frame, $\omega\in\mathbb{R}^3$ is the angular velocity in body frame, $J\in\mathbb{R}^{3\times3}$ is the inertia matrix with respect to rigid body's mass center, $M(t)\in\mathbb{R}^3$ is the torque expressed in the body frame  (control input). 
 
 This paper investigates how to design a feedback controller $M$ that can ensure a \textit{global uniform finite/fixed-time attitude tracking}, i.e.,
$$
\forall R_0 \!\in\! S O(3), \forall \omega_0 \in \mathbb{R}^3, \exists T \!\in\! \mathbb{R}_{+}: R(t)\!=\!R_{\mathrm{d}}(t), \forall t \!\geq\! T,
$$
 where $T$ is independent of $R_0$, $\omega_0$ in the fixed-time case.
For this purpose, this paper adopts the so-called implicit generalized homogeneous algorithms (\cite{Zimenko_etal2020:TAC}) to treat the global finite/fixed time attitude trajectory tracking problem, since they can ensure faster (finite/fixed-time) convergence, better robustness, less overshoots and a simple upgrade of an exists controller to a homogeneous one.

\section{Homogeneous Stabilization}
\label{sec:homo}

In this section, we briefly recall some fundamental concepts of generalized homogeneous control design.

\subsection{Linear dilations in $\R^n$}

By definition, the homogeneity is a dilation symmetry \cite{Zubov1958:IVM}, \cite{Khomenuk1961:IVM}, \cite{Kawski1991:ACDS}, \cite{FischerRuzhansky2016:Book}, \cite{Polyakov2020:Book}, and the dilation (see \cite{Husch1970:Math_Ann,Kawski1991:ACDS}) in a normed vector space is a one-parameter group $\dn(s), s\in \R$ of 
transformations satisfying the limit property $\lim_{s\to \pm\infty}\|\dn(s)x\|=e^{\pm\infty}, \forall x\neq \zero$,  where $\|\cdot\|$ is a norm in $\R^n$. In the general case, the dilation may not be a Lie (continuous) group. 
Examples of  continuous dilations  in $\mathbb{R}^{n}$ are 
\begin{itemize}
\item  Uniform dilation (L. Euler, 18th century):
$
\mathbf{d}(s)=e^{s} I,
$
where $I$ is the identity matrix $\mathbb{R}^{n}$;
\item Weighted dilation (\cite{Zubov1958:IVM}) :
$
\mathbf{d}(s)\!=\!\left[\begin{smallmatrix}\! e^{r_{1} s} & \!...\! & 0\!\\ \!... &\! ...\! & ... \!\\ \!0 & \!...\! & e^{r_{n}\! s}\end{smallmatrix}\right],
$
where $r_{i}>0, i=1,2, \ldots, n$.

\item Geometric dilation (
\cite{Kawski1991:ACDS}) is a flow generated by unstable $C^1$ vector field in $\R^n$.
\end{itemize}
In this paper we deal only with the linear (geometric) dilation in $\R^n$ which is defined as follows \vspace{-2mm}
\begin{equation}\label{eq:dilation}
 \dn(s)=e^{sG_{\dn}}:=\sum_{i=0}^{+\infty}\tfrac{s^iG_{\dn}^i}{i!}, \quad s\in \R,\vspace{-2mm}
\end{equation}
where $G_{\dn}\in \R^{n\times n}$ is an anti-Hurwitz\footnote{A matrix $G_{\dn}\in \R^{n\times n}$ is anti-Hurwitx if $-G_{\dn}$ is  Hurwitz.} matrix being the generator of the dilation $\dn$. A dilation $\dn$ is monotone
if $s\mapsto\|\dn(s)x\|$ is a strictly increasing function for any $x\neq0$.
 It is worth noting that monotonicity of the dilation may depend of the norm $\|\cdot\|$ in $\R^n$. Also, any linear dilation in $\R^n$ is monotone \cite{Polyakov2020:Book} provided that the norm in $\R^n$ is defined as follows\vspace{-2mm}
	\begin{equation}\label{eq:monoticity}
		\|x\|=\sqrt{x^{\top} P x}, \quad PG_{\dn}+G_{\dn}^{\top}P\succ 0, P\succ 0.\vspace{-2mm}
\end{equation}
The linear dilation introduces an alternative norm topology in $\R^n$ by means of a homogeneous norm (see, e.g., \cite{Grune2000:SIAM_JCO} for an example of a homogeneous norm induced by the weighted dilation).

\begin{definition}
	The scalar-valued function $\|\cdot\|_\dn:\mathbb{R}\mapsto [0,\infty)$ defined as $\|\boldsymbol{0}\|_\dn=0$ and \vspace{-2mm}
	$$
	\|u\|_\dn=e^{s_u}, \text{where} \  s_u\in\mathbb{R}:\|\dn(-s_u)x\|=1,\vspace{-2mm}
	$$
	is called the canonical $\dn$-homogeneous norm in $\mathbb{R}^n$, where $\dn$ is a monotone dilation in $\mathbb{R}^n$.
\end{definition}

For any linear monotone dilation in $\mathbb{R}^n$, the canonical homogeneous norm is continuous on $\mathbb{R}^n$ and locally Lipschitz continuous on $\mathbb{R}^n \backslash\{0\}$. Moreover, it is differentiable on $\mathbb{R}^n \backslash\{\boldsymbol{0}\}$ :\vspace{-2mm}
\begin{equation}\label{eq:c_h_n_der}
	\tfrac{\partial\|x\|_\dn}{\partial x}=\|x\|_{\dn} \tfrac{x^{\top} \dn^{\top}\left(-\ln \|x\|_{\dn}\right) P \dn\left(-\ln \|x\|_{\dn}\right)}{x^{\top} \dn^{\top}\left(-\ln \|x\|_{\dn}\right) P G_{\dn} \dn\left(-\ln \|x\|_{\dn}\right) x}\vspace{-2mm}
\end{equation}
 provided that the canonical homogeneous norm is induced by the  Euclidean norm \eqref{eq:monoticity}.
	Below we use the canonical homogeneous norm as an implicit Lyapunov function for an analysis  impulsive locally homogeneous systems and define the $\dn$-homogeneous projector on the unit sphere \cite[p. 159]{Polyakov2020:Book} as follows\vspace{-2mm}
\begin{equation}
	\pi_{\dn}(x)=\dn(-\ln \|x\|_{\dn})x, \quad x\neq \zero.\vspace{-2mm}
\end{equation} 
Indeed, for any $x\neq \zero$, by definition of the canonical homogeneous norm, we have $\|\pi_{\dn}(x)\|=1$.
\subsection{Homogeneous Systems}

\begin{definition}\cite{Kawski1991:ACDS}
A  vector field $f:\R^n\to \R^n$ (resp. a function $h:\R^n \to \R$) is said to be $\dn$-homogeneous of degree $\mu\in \R$ if\vspace{-2mm}
\[
f(\dn(s)x)=e^{\mu s} \dn(s) f(x), (\text{resp. } h(\dn(s)x)=e^{\mu s} h(x)),\vspace{-2mm}
\]
$\forall x\in\R^n, \quad \forall s\in \R$, where $\dn$ is a linear dilation in $\R^n$.
\end{definition}

Any $\dn$-homogeneous system\vspace{-2mm}
\begin{equation}\label{eq:f(x)}
 \dot x=f(x), \quad t>0, \quad x(0)=x_0\in \R^n\vspace{-2mm}
\end{equation}
is diffeomorphic on $\mathbb{R}^n \backslash\{\boldsymbol{0}\}$ and homeomorphic on $\R^n$ to a standard homogeneous system (\cite{Polyakov2020:Book}). This means that many important results known for standard and weighted homogeneous systems hold for linear homogeneous systems as well. An important feature of a stable homogeneous system is a dependence of the convergence rate on the homogeneity degree. The following result is the straightforward corollary of the Zubov-Rosier Theorem (\cite{Zubov1958:IVM}, \cite{Rosier1992:SCL}).

\begin{theorem}\label{thm:exist_hom_LF}
       Let a continuous vector field $f: \mathbb{R}^n \mapsto \mathbb{R}^n$ be a d-homogeneous of degree $\mu \in \mathbb{R}$. The system \eqref{eq:f(x)} is globally uniformly asymptotically stable if and only if there exists a positive definite $\dn$-homogeneous function $V: \mathbb{R}^n \mapsto[0,+\infty)$ such that $V \in C\left(\mathbb{R}^n\right) \cap C^1\left(\mathbb{R}^n \backslash\{\zero\}\right)$,
      $
      \dot V(x)\leq -\rho V^{1+\mu}(x), \quad \forall x\neq 0,
      $
       $\rho>0$. Moreover, the  system \eqref{eq:f(x)} is 
       globally uniformly finite-time stable\footnote{The system \eqref{eq:f(x)} is finite-time stable it is Lyapunov stable and $\exists T(x_0): \|x(t)\|=0, \forall t\geq T(x_0), \forall x_0\in \R^n$.} for $\mu<0$;
      	 globally uniformly exponentially stable for $\mu=0$; 
      	 globally uniformly nearly fixed-time stable\footnote{The system \eqref{eq:f(x)} is uniformly nearly fixed-time stable it is Lyapunov stable and $\forall r>0, \exists T_r>0: \|x(t)\|<r, \forall t\geq T_r$ independently of $x_0\in \R^n$.} for $\mu\!>\!0$.   
 \end{theorem}

\section{Main results}
\label{sec:resu}
In  this section we design a global homogeneous attitude control. First, 
we study the problem of  homogeneous stabilization of a class of impulsive systems.
Then we apply the obtained results to design global homogeneous controller for an attitude control of rigid body, based on the fact that 
 system \eqref{eq:dyn} can be re-formulated as a impulsive system on Lie algebra   $\mathfrak{so}(3)$.

\subsection{ Homogeneous Lyapunov function for  Impulsive Nonlinear Systems}
\label{subsec:lis}
Let us consider the following impulsive system \vspace{-2mm}
\begin{equation}\label{eq:hybrid_linear}
\left\{
	\begin{aligned}
		\dot x
		&=	 f(t, x),\quad \;\;\;\,x^-\in\mathcal{C},\\
		x &=	\Pi 	x^-,\quad \;\;\;\ \ \ x^-\in\mathcal{D},\\
	\end{aligned}
	\right.\vspace{-2mm}
\end{equation}
where $t>t_0$, $x\in\mathbb{R}^n$,  $f \in C(\R^{n+1} ,\R^{n})$, 
the matrix $\Pi\in \R^{n\times n}$ defines a linear reset map, 
$ 
x^-=\lim_{h\to 0^-}x(t+h),
$ \vspace{-2mm}
\begin{equation}\label{eq:CD}
	\mathcal{C}\!=\!\{x\!\in\! \R^n\!:\! g_{1}(x)\!<\!0\}, \mathcal{D}\!=\!\left\{x\!\in\! \R^n\!:\! \begin{smallmatrix}g_{1}(x)=0,\\ g_{2}(x)>0\end{smallmatrix}\right\}\vspace{-2mm}
\end{equation}
 are subsets of $\R^n$, $g_1,g_2\in C^1(\R^n,\R)$. Inspired by \cite{Hespanha_etal2008:Aut}, a solution of the system is a right-continuous function $x$ satisfying the differential equation   almost everywhere, belonging to $\mathcal{C}$ almost everywhere  and  having left limits at each $t>t_0$. 
The following theorem provides a sufficient condition for existence of solutions and  a finite/fixed-time stability analysis of the impulsive system \eqref{eq:hybrid_linear}  by means of implicit homogeneous Lyapunov function designed for the impulse-free system.

\begin{theorem}\label{thm:hybrid_hom_stabilization}
	Let $C$  given by \eqref{eq:CD} be a open connected set such that $\boldsymbol{0}\in\mathcal{C}$,    $\Pi \mathcal{D}\subset\partial C \backslash \overline{D}$,   
		\begin{equation}\label{eq:cond1_thm}
		\tfrac{\partial g_1}{\partial x}f(t,x)<0, \quad \forall x\in \partial C\backslash\overline{\mathcal{D}} , \quad \forall t\in \R	
		\end{equation}
		and for any $r>0$ there exists $\gamma>0$ such that 
		\begin{equation} \label{eq:cond2_thm}
				\tfrac{\partial g_1}{\partial x}f(t,x)\!\leq\! 0,	\; \tfrac{\partial g_2}{\partial x}f(t,x)\!\leq\! -\gamma, \; \forall x\!\in\! \partial D\cap B^n[r], \; \forall t\!\in\! \R.
		\end{equation}	
	 Let the canonical homogeneous norm  induced by the weighted Euclidean norm \eqref{eq:monoticity}  be a local Lyapunov function $V(x)=\|x\|_{\dn}$ of the impulse-free system  $\dot{x} = f(t,x)$ satisfying \vspace{-2mm}
	 \begin{equation} \label{eq:condition_1}
	 \frac{\partial V}{\partial x}f(t,x)\leq -\rho V^{1+\mu}(x),\quad  \forall x\in \mathcal{C}\backslash\{\zero\}, \forall t\in \R,\vspace{-2mm}
	 \end{equation}
	 where $ \mu \in \R$ and $\rho>0$. If \vspace{-2mm}
	 \begin{equation}
		G_\dn\Pi=\Pi G_\dn,  \label{eq:hybrid_condition_a}\vspace{-2mm}
		\end{equation}
	\begin{equation}
		y^{\top}\Pi^{\top}P\Pi y\le 1, \ \forall y \in  \pi_{\dn}(\mathcal{D}) \label{eq:hybrid_condition_b}\vspace{-1mm}
	\end{equation}
then the closed-loop impulsive system \eqref{eq:hybrid_linear} is\vspace{-2mm}
	\begin{itemize} 
		\item uniformly finite-time stable for $\mu<0$ and\vspace{-2mm}
		\[
		x(t)=\zero, \; \forall t\geq t_0+\tfrac{\|x(t_0)\|_{\dn}^{-\mu}}{-\rho \mu},  \forall x(t_0)\in  \overline{\mathcal{C}}, \forall t_0\in \R;\vspace{-2mm}
		\]
		\item uniformly exponentially stable for $\mu=0$ and\vspace{-2mm}
		\[
		\|x(t)\|_{\dn}\leq \tfrac{\|x(t_0)\|_{\dn}}{e^{\rho (t-t_0)}}, \; \forall t\geq t_0, \forall x(t_0)\in  \overline{\mathcal{C}}, \forall t_0\in \R; \vspace{-2mm}
		\]	 
		\item uniformly nearly fixed-time stable for $\mu>0$,\vspace{-2mm}
		\[
		\|x(t)\|_{\dn}\!\leq\! r, \; \forall t\!\geq\! t_0+\tfrac{1}{\rho r^{\mu}}, \forall r\!>\!0, \forall x(t_0)\!\in\!\overline{\mathcal{C}},\forall t_0\!\in\! \R. \vspace{-2mm}
		\] 
	\end{itemize}
\end{theorem}
\begin{proof}
  Since $f\in C(\R^{n+1} ,\R^{n})$,  $\mathcal{C}$ is an open connected set  and $\mathcal{D}\subset \partial \mathcal{C}$,  then  for any initial value $x(t_0)\in \mathcal{C}, t_0\in \R$  the impulsive system \eqref{eq:hybrid_linear} has a solution defined as long as $x^-\notin \partial \mathcal{C}$ and $|x^-|<+\infty$. 
The condition \eqref{eq:cond1_thm}  implies that for  any $t\in \R$ and any $y\in \partial \mathcal{C}\backslash\overline{\mathcal{D}}$, the vector field $f(t,y)$ is oriented inside the open connected set $\mathcal{C}$, so  $\partial \mathcal{C}\backslash\mathcal{D}$ is a repealing set of the system \eqref{eq:hybrid_linear}. Hence, for any initial value $x(t_0)\in \overline{\mathcal{C}}\backslash\overline{\mathcal{D}}$ the system $\mathcal{H}$ has a solution $x(t)$  (belonging to $\mathcal{C}$ for $t>t_0$) as long as 
  $x^-\notin \overline{\mathcal{D}}$ and $|x^-|<+\infty$.  
  Similarly, the condition \eqref{eq:cond2_thm} guarantees that for any $y\in \partial \mathcal{D}$  the vector field $f(t,y)$ is oriented inside the set $\overline{\mathcal{C}}\backslash\{\overline{\mathcal{D}}\}$. So, 
  for any initial value $x(t_0)\in \overline{\mathcal{C}}\backslash\mathcal{D}$ the system \eqref{eq:hybrid_linear} has a solution $x(t)$  (belonging to $\mathcal{C}$ for $t>t_0$) as long as 
  $x^-\notin\mathcal{D}$ and $|x^-|<+\infty$.  
  Since $\Pi\mathcal{D}\subset  \partial C \backslash \overline{D}$  then $\Pi\mathcal{D}\cap \mathcal{D}=\emptyset$ and  there are no multiple values at the instant of jump and  the solution enters in $\mathcal{C}$ after the jump, so it can be continuously prolonged till the next jump.   Let $t_k$ with $k=1,2,...$ denote time instances such that $x^-(t_k)\in \mathcal{D}$. Let us  prove 
   that the sequence $t_k$  does not have an accumulation point. Suppose the contrary, $t_k\to t^*\neq +\infty$ as $k\to +\infty$ provided that $x(t)$ is uniformly bounded $\|x(t)\|<r$ on $[t_0,t^*]$. Since the solution is continuous and uniformly bounded  between  jumps then, due to continuity of $f$,  its derivative is uniformly bounded $\|\dot x(t)\|\leq p,\forall t\in (t_{k},t_{k+1})$, where 
  $ p=\max_{t\in [t_0,t^*],x\in B^n[r])}\|f(t,x)\|.
  $ Then $\delta_k=x^-(t_{k+1})-x(t_{k})=\int^{t_{k+1}}_{t_k}\dot x(\tau) d\tau \to \zero$ as $k\to+\infty$. 
  By Bolzano-Weierstrass Theorem there  always exists a subsequence $t_{k_i}$ such that $x^-(t_{k_i})\to x^*\in \overline{\mathcal{D}}\cap B^n[r]$ as $i\to+\infty$. 
  In this case, $\Pi x^-(t_{k_i})\to \Pi x^*$ as $i\to +\infty$.
	Since  $\Pi x^-(t_{k_i})=x^-(t_{k_{i}+1})-\delta_{k_{i}}$ and $\delta_{k_{i}}\to \zero$ as $i\to+\infty$ 
	then $\Pi x^*=x^{**}:=\lim_{i\to+\infty} x^-(t_{k_{i}}^-)\in \overline{\mathcal{D}}\cap B^n[r]$. Notice that $x^*\in \partial \mathcal{D}\cap B^n[r]$,  otherwise $\Pi x^*=x^{**}\in \partial C \backslash \overline{D}$, but this contradicts to 	$x^{**}\in \overline{\mathcal{D}}$. 
	The condition  \eqref{eq:cond2_thm} together with the continuity of $f$ and the smoothness of $g_2$ imply that 
	for a sufficiently small  $\delta>0$ we have $\tfrac{\partial g_2}{\partial x}f(t,x)\leq-\gamma/2$ for all $x\in 
	(\partial D\cap B^n[r])\dot +B^{n}[\delta]$, there $\dot +$ denotes the geometric sum of sets. 
	Moreover,  there exists $i_r\in \mathbb{N}$ such that for all $i\geq i_r$ 	we have $x(t)\in x^*\dot +B^{n}[\delta],\forall t\in (t_{k_i},t_{k_i+1}).$ Taking into account $g_2(x(t_{k_i})<0$ we derive $g_2(x^-(t_{k_i+1}))=g_2(x(t_{k_i})+\int^{t_{k_i+1}}_{t_{k_i}} \tfrac{\partial g_2}{\partial x}\dot x dt<0$. This contradicts to $x^-(t_{k_i+1})\in \mathcal{D}$. Therefore, 
	any initial value $x(t_0)\in \overline{\mathcal{C}}, t_0\in \R$ the system \eqref{eq:hybrid_linear} has a solution
	defined  as long as $\|x^-\|<+\infty$.
	
For any $x(t_0)\in \mathcal{C}$, due to the condition \eqref{eq:condition_1} and the above explanations, the function $t\mapsto V(x(t))$ decreases for $t>t_0$ as long as $x^-\notin  \mathcal{D}$.		
	On the one hand, 
	 since $\dn(s) = e^{G_{\dn}s} = \sum_{k = 0}^{\infty}\tfrac{s^kG_{\dn}^k}{k!}$, then the identity \eqref{eq:hybrid_condition_a} implies that:\vspace{-2mm}
	\begin{equation}\label{eq:iden_eq}
		\dn(s)\Pi = \Pi\dn(s), \quad \forall s\in\mathbb{R}.\vspace{-2mm}
\end{equation}
	 On the other hand, the Lyapunov function $V(x)=\|x\|_{\dn}$ decreases during the jump if $\|x\|_\dn\le\|x^-\|_\dn$. The latter is equivalent to $\frac{\|x\|_{\dn}}{\|x^-\|_{\dn}}=\|\dn(-\ln \|x^-\|_{\dn})x\|_{\dn}\leq 1$ and to \vspace{-2mm}
	\begin{equation}\label{eq:ineq1}
		x^{\top}\dn^{\top}(-\ln\|x^-\|_\dn)P\dn(-\ln\|x^-\|_\dn)x\leq 1. \vspace{-2mm}
	\end{equation}
Denoting   $s^-=\ln\|x^-\|_\dn$ and using \eqref{eq:iden_eq}, we derive\vspace{-2mm}
	\begin{equation}\label{eq:ident_thm}
		\begin{aligned}
			&x^{\top}\dn^{\top}\!\left(-s^-\right)\!P\dn\!\left(-s^-\right)\!x\! \\
			=&
			(x^-)^{\top}\Pi^\top\dn^{\top}\!\left(-s^-\right)\!P\dn\!\left(-s^-\right)\!\Pi x^-\\
			=&
			\left(x^-\right)^{\top}
			\dn^{\top}\left(-s^-\right)\Pi^{\top}P\Pi\dn\left(-s^-\right)
			x^-	, \quad  x^-\!\in\!\mathcal{D}.
		\end{aligned}\vspace{-2mm}
	\end{equation}
		Since $\|\dn(-s^-)x^-\|_{\dn}=1$ then $\dn(-s^-)x^-\in \pi_{\dn}(\mathcal{D})$, where $\pi_{\dn}(\mathcal{D})$ is the $\dn$-homogeneous projection of the set $\mathcal{D}$ on the unit sphere. It is well defined since $\zero \notin \mathcal{D}$. 
	Therefore, the condition \eqref{eq:hybrid_condition_b}  and the identity \eqref{eq:ident_thm} imply the inequality \eqref{eq:ineq1}, or equivalently, $\|x\|_\dn\le\|x^-\|_\dn$ for any $x^-\in \mathcal{D}$.

	The latter means that the  function $ t\mapsto V(x(t))$ is monotone decreasing along any trajectory $x(t)$ of the system and the sequence $t_k$ of time instances (of jumps) has no accumulation points (see the above explanations). In this case,  the inequality \eqref{eq:condition_1}
	holds almost everywhere along the trajectory of the system and the inequality $V(x(t))\leq V(0)-\rho\int^t_0 V^{1+\mu}(x(\tau)) d \tau$ is fulfilled due to monotonicity of $t\mapsto V(x(t))$.
	Hence, we conclude that 
	$(-\mu)^{-1} (V^{-\mu}(x(t)) - V^{-\mu}(x_0))\leq - \rho t$ for $\mu\neq 0$ and $V(x(t))\!\leq\! e^{-\rho t} V(x_0)$ for $\mu\!=\!0$. The proof is complete.
\end{proof}

  Notice that,  the existence of a Lyapunov function in the form of a canonical homogeneous norm is  not a conservative condition for the class of stable $\dn$-homogeneous vector fields  $f$ (see, \cite{Polyakov2020:Book} for more details).  Theorem  \ref{thm:hybrid_hom_stabilization} can also be utilized for non-homogeneous vector fields, which are close to homogeneous in some sense.
Below, we investigate the stability of a specific  impulsive system, which is common in practice of attitude control. 
\begin{remark}\label{cor:hybrid_stability}
 If $x\!\in\!\mathbb{R}^{2n}$, $u\!\in\!\mathbb{R}^n$,  $G_\dn=\text{diag}(g_1 I_n, g_2  I_n)$, $g_1>0,g_2>0$, $P=\left[\begin{smallmatrix}
			p_{11}I_n & p_{12} I_n\\
			p_{12}I_n & p_{22}I_n
		\end{smallmatrix}\right]$, $p_{i,j}>0$, $\forall i,j\in\mathbb{N}_{\le 2}$, $\Pi\!=\!\left[\begin{smallmatrix}
		-I_n & \boldsymbol{0}\\
		\boldsymbol{0} & I_n
	\end{smallmatrix}\right]$,
	$\mathcal{C}=\{ x\in\mathbb{R}^{2n}: x^{\top}\!H^{\top}\!H
	x<r\}$, 
	$\mathcal{D}\!=\!\left\{x\in\mathbb{R}^{2n}: x^{\top}\!H^{\top}\!H
	x=r, x^{\top}Qx> 0 \right\},$
	where $r>0$ is a constant, and $Q\!=\!\left[\begin{smallmatrix}
		\boldsymbol{0} & I_n\\
		I_n & \boldsymbol{0}
	\end{smallmatrix}\right]$ and $H\in \R^n$ is an arbitrary non-zero matrix,
then $\Pi \mathcal{D}\subset\partial C \backslash \overline{D}$ 	 
	and the conditions \eqref{eq:hybrid_condition_a}, \eqref{eq:hybrid_condition_b} are fulfilled.
\end{remark}
\begin{proof} The identity $\Pi^\top Q \Pi = -Q$ implies that $\Pi \mathcal{D}\subset\partial C \backslash \overline{D}$.		
	On the one hand, for $G_\dn=\text{diag}(g_1 I_n, g_2  I_n)$  we have 
$ G_\dn \Pi \!=\! \Pi G_\dn$.
	On the other hand, 
	if $x=[x_1^\top,x_2^\top]^\top \in \mathcal{D}$, then \vspace{-1mm}
	\[
	x^{\!\top}\!Qx\!=2x_2^{\!\top}x_1\!=2x_1^{\!\top}x_2\!> \!0. \vspace{-1mm}
	\]
	In this case, for $p_{11}>0$, $p_{12}>0$, $p_{22}>0$, we derive \vspace{-2mm}
	\begin{equation*}
		\begin{aligned}
			&x^{\top}\Pi^{\top} P \Pi x= p_{11}x_1^{\top}x_1-2p_{12}x_1^{\top}x_2+p_{22}x_2^{\top}x_2\leq\\
			&
			p_{11}x_1^{\top}x_1+2p_{12}x_1^{\top}x_2+p_{22}x_2^{\top}x_2=x^{\top}Px.
		\end{aligned}\vspace{-2mm}
	\end{equation*}
{
Since $y \in \{\pi_{\dn}(x), \forall x\in \mathcal{D}\}$, according to the latter inequality, we have $y^\top\Pi^{\top} P \Pi y \le y^\top P y  $
Then, according to the definition of canonical homogeneous norm $\|x\|_\dn\Rightarrow \|y\|=1$, we have $y^\top\Pi^{\top} P \Pi y \le 1$.}
This means that the condition \eqref{eq:hybrid_condition_b} is fulfilled as well.
\end{proof}


\subsection{Global Homogeneous Attitude Control}

 
The linear dilations, as well as the generalized homogeneity, are introduced for normed vector spaces, therefore the $SO(3)$ group does not admit a dilation in the above sense. So, the homogeneity-based control cannot be introduced  directly  for the model \eqref{eq:dyn1}. However, the Lie algebra $\mathfrak{so}(3)$ is equivalent to  a vector space, so the homogeneous control can be designed for a model  expressed  in terms of the exponential coordinates, which will be discussed below.

\subsubsection{Global Attitude Dynamics on Lie Algebra $\mathfrak{so}(3)$}

 Let us introduce an attitude tracking error \vspace{-2mm}
\begin{equation}\label{eq:err_def}
	\begin{aligned}
	 R_e= RR_d^{\top}=e^{\theta_e^{\wedge}}=\mathrm{Exp}(\theta_e)
	\end{aligned} \vspace{-2mm}
\end{equation}
where $R_d \in C^2\left(\mathbb{R}_{+}, S O(3)\right)$ represents a desired smooth attitude trajectory and $\theta_e \in \mathbb{R}^3$. For a given smooth trajectory, the attitude error dynamics can be derived in the exponential coordinates, and the relation between $\theta_e$ and $\omega_e$ is (see, \eg,  \cite{chirikjian2012information}) \vspace{-2mm}
\begin{equation}
	\dot{\theta}_e= J_{\textit{r}}^{-1}\omega_e, \vspace{-2mm}
\end{equation}
 where $J_{\textit{r}}=J_{\textit{r}}(\theta_e)$ is the right Jacobian defined in \eqref{eq:J_r},  $\omega_e=R_d\left(\omega-\omega_d\right)$   and 
$\omega_d=(R_d^{\top}\dot{R}_d)^{\vee}$ is the desired angular velocity in body frame.

The error $\theta_e\in B^3[\pi]$ can capture all elements of $R_e$, but the ambiguity appears on the sphere $\left\{\theta_e\in \mathbb{R}^3: |\theta_e|=\pi\right\}$.
The map from $SO(3)$ to $\mathfrak{so}(3)$ is a bijection when $\theta_e\in B^3(\pi)$,
to complete the dynamics 
on 
$B^{3}[\pi]\times\mathbb{R}^3$,  a jump mechanism $\theta_e=-\theta_e^-$ is introduced  for $| \theta_e^- |=\pi$ and $\frac{d}{dt}|\theta^-_e|>0$ (see Appendix or \cite{bharadwaj1998geometry} for more details).

Given a smooth attitude trajectory $R_d$, the global error dynamics of \eqref{eq:err_def} can be then formulated as an  impulsive system (\cite{bharadwaj1998geometry})\vspace{-2mm} 
	\begin{equation}\label{eq:att_err}
		\left\{
			\begin{aligned}
				& \dot{\xi}	=\left[\begin{smallmatrix}
					\boldsymbol{0} & J_r^{-1}\\
					\boldsymbol{0} & \boldsymbol{0}
				\end{smallmatrix}\right]	\xi		+	Bu,\quad \xi^-\in\mathcal{C},\\
				&		\xi=	\Pi \xi^-,\ \qquad \qquad \quad\;\, \xi^-\in\mathcal{D},\\		
			\end{aligned}
		\right.\vspace{-2mm}
	\end{equation}
with $\xi=\left[\begin{smallmatrix}\theta_e\\ \omega_e\end{smallmatrix}\right]$, $\mathcal{D}=\left\{\xi \in  { B^{3}[\pi]\times\mathbb{R}^3}: \xi^{\top} H^{\top} H \xi=\pi^2\right.$ and $\left.\xi^{\top} Q \xi>0\right\}$, $\mathcal{C} = { B^{3}(\pi)\times\mathbb{R}^3}$,  where  $ H=\left[\begin{smallmatrix}I_3 & 0 \\ 0 & 0\end{smallmatrix}\right]$, $B = [
	\boldsymbol{0},
	I_3
]^\top$, the matrices $\Pi$ and $Q$ are as in Corollary \ref{cor:hybrid_stability} with $n=3$.
{
The $u =\dot{\omega}_e$ denotes the virtual {stabilizing} control needs to be designed.
}
The system \eqref{eq:dyn} is equivalent to \eqref{eq:att_err}, so  the asymptotic stability of \eqref{eq:att_err} implies $R\rightarrow R_d$ and $\omega\rightarrow\omega_d$ as  $\xi \rightarrow \zero$. 

	

\subsection{Homogeneous attitude control on $\mathfrak{so}(3)$}\label{section:att_control}

As stated in the introduction, there is no continuous time-invariant feedback control for global asymptotic attitude stabilization. The jump behavior of \eqref{eq:att_err} reveals this statement.  
 A continuous feedback  may generate a discontinuous signal due to the state jump. 

\begin{theorem}\label{thm:hybrid_control}
	Let \vspace{-2mm}
	\begin{equation}\label{eq:hom_con_Q}
	u_{hom}\!=\!\|\xi\|_{\dn}^{1+\mu}K\dn(-\ln \|\xi\|_{\dn})\xi, \quad {\xi=\left[\theta_e^\top, \omega_e^\top\right]^\top}\vspace{-2mm}
	\end{equation} be a $\dn$-homogeneous controller with $\mu\in[-1,1)$, $K = [-K_1\;\; -k_2I_3]$, $0\!\prec\! K_1\in \R^{3\times 3}$, $k_2>0,$ the dilation  $\dn(s)=e^{sG_{\dn}}, s\in \R$,  the canonical homogeneous norm induced by the Eclidean norm $\sqrt{\xi^\top P\xi}$ and  \vspace{-2mm}
\begin{equation}\label{eq:P}
\!G_\dn \!=\! \left[\!\begin{smallmatrix}
		(1-\mu)I_3 & \boldsymbol{0}\\
		\boldsymbol{0} & I_3
	\end{smallmatrix}\!\right]\!,  P \!=\! \left[\!\begin{smallmatrix}
		I_3 & \varepsilon I_3\\
		\varepsilon I_3 & K_1^{-1}
	\end{smallmatrix}\!\right]\!, 	 0\!<\!\varepsilon\!<\!\min\!\left\{\epsilon_\mu, \tilde \epsilon\right\}\!,\!\!\vspace{-2mm}
\end{equation} 
 $ \epsilon_{\mu}\!=\!
\!\tfrac{2(1-\mu)^{\frac{1}{2}}}{(2-\mu)\lambda^{\frac{1}{2}}_{\max}\!(K_1)}$, $\tilde \epsilon\!=\! \frac{4k_2}{2c\lambda_{\max }\!\left(K_1\!\right)+k_2^2}
$,  
$c\!=\!\max\limits_{|\theta_e|\leq \pi}\lambda_{\max }(J_r^{-1}\!+\!J_r^{-\top})$.
 Then the
 impulsive system \eqref{eq:att_err}  with $u=u_{hom}$ is:
	\begin{itemize} 
		\item globally uniformly finite-time stable for $\mu<0$:\\
		$
		\xi(t)=\zero, \quad \forall t\geq \|\xi(0)\|^{-\mu}_{\dn}/(-\mu \tilde{\rho});
		$
		\item globally uniformly exponentially stable for $\mu=0$:\\
		$
		\|\xi(t)\|_{\dn}\leq e^{-\tilde{\rho} t}\|\xi(0)\|_{\dn},\quad \forall t\geq 0;
		$		 
		\item globally uniformly nearly fixed-time stable for $\mu>0$:\\ 
		$
		\|\xi(t)\|_{\dn}\leq r,\quad  \forall t\geq \tfrac{1}{\tilde{\rho} \mu r^{\mu}}, \forall r>0;
		$
	\end{itemize}\vspace{-2mm}
 \begin{equation}\label{eq:tilde_rho}
		\tilde{\rho} = \tfrac{\lambda_{\min }\left(P^{-\tfrac{1}{2}}\left[\begin{smallmatrix}
			2\varepsilon K_1 & \varepsilon K_2\\
			\varepsilon K_2 &  \! K_1^{-1}K_2+K_2K_1^{-1}-\varepsilon cI_3\!
		\end{smallmatrix}\right]P^{-\tfrac{1}{2}}\right)}{\lambda_{\max }(P^{1/2}G_\dn P^{-1/2} +P^{-1/2}G_\dn^\top P^{1/2} )}.
\end{equation}
Moreover, the selection \vspace{-2mm}
\begin{equation}\label{eq:att_control}
	M\!=\!J
	\left[ R_d^{\top}u_{hom}\! -\omega_d^{\wedge}\omega\!+\!\dot{\omega}_d\right]			 +\omega\!\times\! J\omega \vspace{-2mm}
\end{equation} 
where 
$		\dot{\omega}_d\!=\!\tfrac{d}{dt}(R_d^{\top}\dot{R}_d)^{\vee}\!=\!(R_d^{\top}\ddot{R}_d\!-\!(\omega_d^{\wedge})^2)^{\vee}\in C(\R_+,\R^3)$,   solves the finite-time attitude tracking problem for the system \eqref{eq:dyn1}, \eqref{eq:dyn2} if $\mu<0$.
\end{theorem}
\begin{proof}
 For $\omega_e = R_d(\omega - \omega_d)$ we derive \vspace{-2mm}
\begin{equation}
		 \dot{\omega}_e=
		R_d\omega_d^\wedge(\omega-\omega_d)
		+R_d(\dot{\omega}-\dot{\omega}_d).\vspace{-2mm}
\end{equation}
		Since $\omega ^\wedge_{d} \omega_d = \zero $ then using 
\eqref{eq:dyn2} and \eqref{eq:att_control} we obtain $\dot {\omega}_{e}=u_{hom}$ and the dynamics of the closed-loop system becomes \eqref{eq:att_err} with $u=u_{hom}$.

The dilation  $\dn$ is monotone since 
$PG_\dn \!+\! G_\dn^\top P\succ 0\!\Leftrightarrow\! 2(1-\mu)I_3\! -\! \tfrac{1}{2}\varepsilon^2(2\!-\!\mu)^2 K_1\succ 0,$ { $P\succ 0 \Leftrightarrow K_1^{-1}\succ \varepsilon^2$}.
The latter inequalities hold for  $\varepsilon<\epsilon_\mu.$
	{ For the Lyapunov function candidate $V(\xi) = \|\xi\|_\dn$}, we have\vspace{-2mm} 
\begin{equation}
	\tfrac{d}{dt}\|\xi\|_\dn \!=\!  \|\xi\|_\dn\tfrac{\xi^\top\dn^\top(-s_\xi)P\dn(-s_\xi)\dot{\xi}}{\xi^\top\dn^\top(-s_\xi)PG_\dn\dn(-s_\xi)\xi}, \quad \forall  \xi^- \!\in\!\mathcal{C}.\vspace{-2mm}
\end{equation}
Hence, using the representation \vspace{-2mm}
\begin{equation*}
	\begin{aligned}
		&\xi^\top \dn^\top(\!-\!s_\xi) P \dn(\!-\!s_\xi)\dot{\xi}
		=
		\left[\!\begin{smallmatrix}
			\|\xi\|_{\dn}^{\mu-1}\theta_e\\
			\|\xi\|_{\dn}^{-1}\omega_e
		\end{smallmatrix}\!\right]^\top \!P\!
		\left[\!\begin{smallmatrix}
			\|\xi\|_{\dn}^{\mu-1}\dot{\theta}_e\\
			\|\xi\|_{\dn}^{-1}\dot{\omega}_e
		\end{smallmatrix}\!\right]\\
		& = 
		\|\xi\|_{\dn}^{2(\mu-1)}\theta_e^\top J_r^{-1}\omega_e +
		\varepsilon \|\xi\|_{\dn}^{\mu-2}\omega_e^\top J_r^{-1}\omega_e\\
		& \quad+
		 \varepsilon \theta_e^\top  \left(-\|\xi\|_\dn^{3\mu-2}K_1 \theta_{e}  - \|\xi\|_\dn^{2(\mu-1)}k_2 \omega_e\right)\\
		&\quad  +  \omega_e^\top K_1^{-1}\left(-\|\xi\|_\dn^{2\mu-2}K_1 \theta_{e}  - \|\xi\|_\dn^{\mu-2}k_2 \omega_e\right)\\
		& \!= 
		\!-\!\varepsilon \|\xi\|_\dn^{3\mu-2}\theta_e^\top K_1 \theta_{e} \!+\! \|\xi\|_\dn^{\mu-2}\omega_e^\top( \varepsilon J_r^{-1}\! -\! k_2K_1^{-1}) \omega_e\\
		&\quad+ \|\xi\|_\dn^{2(\mu-1)}\theta_e^\top(J_r^{-1} - (\varepsilon k_2 +1) I_3) \omega_e 
	\end{aligned}\vspace{-2mm}
\end{equation*}
where $s_\xi = \ln\|\xi\|_\dn$, we derive \vspace{-2mm}
\begin{equation}
	\tfrac{d}{dt}\|\xi\|_\dn \!=\!  \|\xi\|_\dn^{1+\mu}\tfrac{\xi^\top\dn^\top(-s_\xi)W_0\dn(-s_\xi)\xi}{\xi^\top\dn^\top(-s_\xi)PG_\dn\dn(-s_\xi)\xi}, \quad  \xi^- \!\in\!\mathcal{C}\vspace{-2mm}
\end{equation}
where 
$
W_0=\left[
\begin{smallmatrix}
	-\varepsilon K_1 & \frac{1}{2}( J_r^{-1} - (\varepsilon k_2+1)I_3  )\\
	* & \varepsilon \frac{J_r^{-1}+J_r^{-\top}}{2} - k_2K_1^{-1}
\end{smallmatrix}\right].
$
Taking into account the form of $J_r^{-1}$ given by \eqref{eq:J_r} we conclude \vspace{-2mm}
\begin{equation}
	\theta_e^{\top}J_r^{-1}\omega_e =\theta_e^{\top}\omega_e.\vspace{-2mm} 
\end{equation}
Since $J_r^{-1}$ is bounded  for  $|\theta_e|\leq \pi$, then \vspace{-2mm}
\begin{equation}
	\tfrac{d}{dt}\|\xi\|_\dn \!\le\! \tfrac{ \|\xi\|_\dn^{1+\mu} \xi^\top\dn^\top(-s_\xi)
		W
		\dn(-s_\xi)\xi}{\xi^\top\dn^\top(-s_\xi)(PG_\dn + G_\dn^\top P)\dn(-s_\xi)\xi}, \quad  \xi^- \!\in\!\mathcal{C}, \vspace{-2mm}
\end{equation}
where $W \!=\! \left[\begin{smallmatrix}
	-2\varepsilon K_1 & -\varepsilon k_2I_3\\
	-\varepsilon k_2I_3 & \varepsilon cI_3\! -\! 2k_2K_1^{-1}
\end{smallmatrix}\right]$ and $c$ is defined in the statement of the theorem.
Using Schur Complement we derive $W\prec 0$ if 
$
\varepsilon<\tilde{\epsilon}.
$ Therefore, for $
\varepsilon<\min\left\{\epsilon_\mu, \tilde \epsilon\right\},
$ we have $P\succ 0, W\prec 0$. Using to the definition of the canonical homogeneous norm, we obtain $\xi^\top\dn^\top(-s_\xi)P\dn(-s_\xi)\xi=1$ and \vspace{-2mm}
\begin{equation}
	\dot{V}\le-\tilde{\rho} \|\xi\|_\dn^{1+\mu},\quad \quad \xi^-\in\mathcal{C}\backslash\{\zero\}.\vspace{-2mm}
\end{equation}
In the view of Remark \ref{cor:hybrid_stability} to complete the prove by applying Theorem \ref{thm:hybrid_hom_stabilization} we just need to show that conditions \eqref{eq:cond1_thm} and  \eqref{eq:cond2_thm} are fulfilled for $g_1(\xi)=\theta_e^{\top} \theta_e-\pi^2$ and $g_2(\xi)=\omega_e^{\top}\theta_e$. 
Indeed, for $\xi\in \partial C\backslash\overline{\mathcal{D}}$ we have $g_2(\xi)<0$ and\vspace{-2mm}
\[
\tfrac{\partial g_1}{\partial \xi}f(\xi,t) = 2\theta_e^\top J_r^{-1}\omega_e = 2\theta_e^\top\omega_e<0,\ \xi\in \partial C\backslash\overline{\mathcal{D}}\vspace{-2mm}
\]
but for $\xi\in \partial \mathcal{D}$ we have $g_1(\xi) = \theta_e^{\top} \theta_e-\pi^2 = 0$, $g_2(\xi)=\omega_e^{\top}\theta_e = 0$, and \vspace{-2mm}
\[
\tfrac{\partial g_1}{\partial \xi}f(\xi,t) = 2\theta_e^\top J_r^{-1}\omega_e = 2\theta_e^\top\omega_e=0,\ \xi\in \partial \mathcal{D},\vspace{-2mm}
\]
\[
\tfrac{\partial g_2}{\partial \xi}f(\xi,t) \!=\! 2\omega_e^\top J_r^{-1}\omega_e \!+\! 2\|\xi\|_\dn^{\mu+1}\theta_{e}^\top K\dn(-s_\xi)\xi, \ \xi\in \partial \mathcal{D} .
\]
For  $\theta^\top \omega_e = 0$, using \eqref{eq:J_r}, we derive  $\omega_e^\top J_r^{-1}\omega_e=|\omega_e|^2+\omega_e^\top \left(\frac{\theta_e^{\wedge}}{|\theta_e|}\right)^2\omega_e=\omega_e^\top \frac{\theta_e\theta ^{\top}_e}{|\theta_e|^2}\omega_e=0$. Since  $\theta_{e}^\top K\dn(-s_\xi)\xi = -\|\xi\|_\dn^{\mu-1}\theta_e^\top K_1\theta_e -\|\xi\|_\dn^{-1}k_2\theta_e^\top \omega_e$
then \vspace{-2mm}
	\[
	\tfrac{\partial g_2}{\partial \xi}f(\xi,t) = -\|\xi\|_\dn^{2\mu}\theta_{e}^\top K_1\theta_{e}<0,\quad  \xi\in \partial \mathcal{D}.\vspace{-2mm}
	\]
Taking into account the equivalence of Euclidean norm and the canonical homogeneous norm, we conclude that 
for any $\forall r\in\mathbb{R}_+$ there exists $\gamma>0$ such that\vspace{-2mm}
	\[
	\tfrac{\partial g_2}{\partial \xi}f(\xi,t)< -\gamma,  \ \xi\in\! \partial D\cap B^n[r].\vspace{-2mm}
	\]
The proof is complete.
\end{proof} 

Due to the monotone decay of the Lyapunov function, the states of the closed-loop system  \eqref{eq:att_control} and \eqref{eq:att_err} will eventually enter an invariant set being a neighborhood of the origin.
For any bounded initial condition, the system trajectory state will have a finite number of jumps (see the proof of Theorem \ref{thm:hybrid_hom_stabilization}).
Clearly, for $\mu<0$, finite-time convergence to origin is guaranteed but the settling time depends on the initial states; for $\mu>0$ the algorithm ensures a fixed-time convergence to a neighborhood of the origin independently of the initial state of the system. 

 Using a commutation of homogeneous controller (as in \cite{Polyakov_etal2015:Aut}) a global fixed-time stabilization of the tracking error at zero can be guaranteed. 

 \begin{corollary}[Fixed-time stabilization]\label{cor:fxt}
 	Let $-1\leq \mu_2<0< \mu_1<1$, { $K = [-K_1\;\; -k_2I_3]$, $0\!\prec\! K_1\in \R^{3\times 3}, k_2\in\mathbb{R}_+ $}and \vspace{-2mm}
 	\begin{equation}\label{eq:fxt}
 		u_{fxt} = \left\{
 		\begin{array}{ccc}
 			\|\xi\|_{\dn_1}^{1+\mu_1}K\dn(-\ln \|\xi\|_{\dn_1})\xi  &\text{ if }&\xi^{\top}\!P\xi\!\ge\! 1,\\
 			\|\xi\|_{\dn_2}^{1+\mu_2}K\dn(-\ln \|\xi\|_{\dn_2})\xi &\text{ if } &\xi^{\top}\!P\xi\!\le\! 1,
 		 \end{array}
 		\right.\vspace{-2mm}
 	\end{equation}
 where the dilations $\dn_1$, $\dn_2$ and  the canonical homogeneous norms $\|\cdot\|_{\dn_1}, \|\cdot\|_{\dn_2}$ are defined as in Theorem \ref{thm:hybrid_control}  for $\mu=\mu_1$ and $\mu=\mu_2$, respectively, but the matrix $P$ is 
 as in Theorem \ref{thm:hybrid_control} with  $\varepsilon\in (0,\min\{\epsilon_{\mu_1},\epsilon_{\mu_2},\tilde \epsilon\})$. Then the impulsive system  \eqref{eq:att_err} with $u=u_{fxt}$ is globally uniformly fixed-time stable:\vspace{-2mm}
  	$$\forall\xi(t)=\zero, \quad \forall t\ge \tfrac{1}{-\tilde{\rho}_2\mu_2} + \tfrac{1}{\tilde{\rho}_1\mu_1},\vspace{-2mm}$$
  	where $\tilde \rho_1>0$ and $\tilde \rho_2>0$ are defined by the formula \eqref{eq:tilde_rho} with $\mu=\mu_1$ and $\mu=\mu_2$, respectively. 
  	
  	Moreover, the selection of $M$ in the form \eqref{eq:att_control} replacing  $u_{hom}$ with $u_{fxt}$ solves
  	the fixed-time attitude tracking problem for the system \eqref{eq:dyn1}, \eqref{eq:dyn2}.
 \end{corollary}
 
\begin{proof}
 Since  $\xi^{\top}P\xi =1$ (resp., $\leq 1$, $\geq 1$) is equivalent to $\|x\|_{\dn_1}=1$ (resp., $\leq 1, \geq 1$)
		and it is  equivalent to $\|x\|_{\dn_2}=1$ (resp., $\leq 1, \geq 1$), then, in both cases $\mu=\mu_1$ and  $\mu=\mu_2$, 
	the unit ball $\{\xi: \xi^{\top}P\xi \leq 1\}$ is a strictly positively invariant and globally uniformly  attractive set of   the system \eqref{eq:att_err} with the homogeneous control \eqref{eq:hom_con_Q}. So,	due to the estimate $\dot V\leq -\tilde \rho V^{1+\mu}$ (see, the proof of Theorem \ref{thm:hybrid_control}),
	 any trajectory of the closed-loop system \eqref{eq:att_err}, \eqref{eq:fxt} converges to the unit ball in a fixed time, namely,
	$\|\xi(t)\|_{\dn_2}\leq 1,\  \forall t\ge T_1 = \tfrac{1}{\tilde{\rho}\mu_1}$.  The latter means that $\xi^{\top}(t)\!P\xi(t)\!\le\! 1$ for all $t\geq T_1$ (i.e., $\mu=\mu_2$), so the system converges to zero in a fixed time and  $\xi(t)=\zero$ for all  $t\!\ge\! T_1\!+\!\tfrac{1}{-\tilde{\rho}\mu_2}$ independently of the initial state $\xi(0)$.  
\end{proof}

The switching control \eqref{eq:fxt} combines two homogeneous controllers \cite{Polyakov_etal2015:Aut} to guarantee a global fixed-time tracking. Since 
$\|\xi\|_{\dn_1}=\|\xi\|_{\dn_2}=1$ for $\xi^{\top}P\xi=1$ and  $\dn_1(0)=\dn_2(0)=I_6$, then the control $u_{fxt}$ is  a continuous on $\{\xi: \xi^{\top}P\xi=1\}$.

\subsection{Robustness analysis of homogeneous controller}
Any control system is subjected to perturbations and measurement noises in practice. In this case, it is important to investigate its robustness in the sense of Input-to-State Stability (ISS), \cite{Sontag1989:TAC}, \cite{Hespanha_etal2008:Aut}. Below we also prove the so-called finite-time and fixed-time ISS developed in  \cite{Hong_etal2010:SIAM}, \cite{Lopez-Ramirez_etal2020:SCL}, \cite{Aleksandrov_etal2022:MCSS}.
Notice that the noised measurement  of $R\in SO(3)$  can be modeled as follows 
$$\tilde{R}=d_R R,$$ where $\tilde R\in SO(3)$ is a measured/estimated rotation matrix, and $d_R\in SO(3)$ models measurement noise on $SO(3)$.
In this case, we derive \vspace{-2mm} $$\tilde{\theta}_e^{\wedge}=
\log(d_R R R^{\top}_d)=\log(R_e)+\log(d_R) =\theta_e^{\wedge}+\log(d_R),\vspace{-2mm}$$
i.e. the model of measurements in exponential coordinates (on the Lie algebra $\mathfrak{so}(3)$) includes the noise in the additive manner. So, the perturbed closed-loop attitude control system can be modeled as follows	\vspace{-2mm}  
	\begin{equation}\label{eq:att_err_pert}
		\left\{
		\begin{aligned}
			& 	\dot{\xi}
			\!=\!
			 \tilde{A}
			\xi
			\!+\!
			Bu(\xi+\delta_1) \!+\!\delta_2,  \ &\xi^-\in\mathcal{C},\\
			&		\xi=	H	\xi^-,& \xi^-\in\mathcal{D},\\				
		\end{aligned}
		\right.\vspace{-2mm}
	\end{equation}
where $\tilde{A} = \left[\begin{smallmatrix}
	\boldsymbol{0} & J_r^{-1}\\
	\boldsymbol{0} & \boldsymbol{0}
\end{smallmatrix}\right]$, $\xi,u$, $\mathcal{C},\mathcal{D}$ are as before,  $\delta_1 = [(\log(d_R)^\vee)^\top, d_\omega^\top]^\top\!\in\mathbb{R}^6$ is a measurement noise, $\delta_2\in\mathbb{R}^6$ is an exogenous perturbation and   $d_\omega\in\mathbb{R}^3$ is the measurement noise of the angular velocity.

\begin{corollary}\label{thm:coro1}
Under conditions of Theorem \ref{thm:hybrid_control} (Corollary \ref{cor:fxt}), the impulsive system \eqref{eq:att_err_pert} with $u\!=\!u_{hom}$ ($u\!=\!u_{fxt}$) is 
\begin{itemize} 
	\item locally ISS with respect to $\delta=[\delta_1^{\top},\delta_2^{\top}]^{\top}\in L^{\infty}(\R,\R^6)$,
	\item  (fixed-time)  ISS with respect to $\delta\!=\![\delta_1^{\top},\delta_2^{\top}]^{\!\top}\in \Delta$, 
\end{itemize}	
	for $\mu>-1$,  where $\Delta\subset L^{\infty}(\R,\R^6)$ is set of perturbations such that 
	 the system	\eqref{eq:att_err_pert} is forward complete\footnote{The system is forward complete if all solutions of the system exist and defined globally in the forward time.}.
\end{corollary}

\begin{proof}
	 Let the  vector  field $f:\R^{18} \mapsto  \R^{6}$ be defined as \vspace{-2mm}
	$$f(\xi,\delta)=\dot{\xi}
	=	\tilde{A}\xi	+	Bu_{hom}(\xi+\delta_1) +\delta_2,\vspace{-2mm}$$
	where $\xi\in \R^6,\delta_1,\delta_2\in \R^{6}, \delta=(\delta_1^{\top},\delta_2^{\top})^{\top}\in \R^{12}$.
	For $V = \|\xi\|_\dn$ we have \vspace{-2mm}
	\begin{equation}\label{eq:partial}
		\tfrac{\partial V(z)}{\partial z} \big |_{z = \dn(s)\xi} \dn(s)  = e^s\tfrac{\partial V}{\partial \xi}, \forall \xi\neq \zero, \forall s\in \R .  \vspace{-2mm}
	\end{equation}
In this case, using the identity $\dn(s)B=e^sB$  we derive \vspace{-2mm}
\[
\begin{array}{l}
\tfrac{\partial V}{\partial \xi} f(\xi,\delta) \!=\!\tfrac{\partial V}{\partial \xi} f(\xi,0) \!+\!  \tfrac{\partial V}{\partial \xi} \left(f(\xi,\delta)\!-\!f(\xi,0)\right)=\\
\tfrac{\partial V}{\partial \xi} f(\xi,0)+\tfrac{\partial V}{\partial \xi} \left(Bu_{hom}(\xi+\delta_1)-Bu_{hom}(\xi)+\delta_2\right)=\\
\tfrac{\partial V}{\partial \xi} f(\xi,0)\!+\!\tfrac{\partial V(z)}{\partial z} 
\left(B(u_{hom}(\xi\!+\!\delta_1)-u_{hom}(\xi))\!+\!\dn(s)\delta_2\right)\!,
\end{array}\vspace{-2mm}
\]
where $z=\dn(s)\xi$.  The $\dn$-homogeneity of $u_{hom}$ gives\vspace{-2mm}
$$u_{hom}(y)=e^{-(1+\mu)s}u_{hom}(\dn(s)y), \forall y\in \R^6,\forall s\in \R.\vspace{-2mm}$$ Taking $s=-\ln\|\xi\|_{\dn}$ we derive  
$
\|Bu_{hom}(\xi+\delta_1)-Bu_{hom}(\xi)\|= \|\xi\|_{\dn}^{1+\mu} \|Bu_{hom}(z+\dn(-\ln \|\xi\|_{\dn})\delta_1)-Bu_{hom}(z)\|,
$
where $z=\dn(-\ln \|\xi\|_{\dn})\xi$ belongs to the unit sphere  $ \|z\|=1$.
Since $u_{hom}\in C^1(\R^{6}\backslash\{\zero\},\R^3)$ then for $\|\dn(-\ln\|\xi\|_{\dn})\delta_1\|\leq 0.5$ 
using mean value inequality we derive\vspace{-2mm}
\[
\|Bu_{hom}(z+\dn(s)\delta_1)-Bu_{hom}(z)\|\leq \vspace{-2mm}
\]
\[ \left\|B\smallint^1_0\left. \tfrac{\partial u_{hom}(y)}{\partial y}\right|_{y=z+\lambda \dn(s)\delta_1} \dn(s)\delta_1 d \lambda \right\|\leq  C_1 \|\dn(s)\delta_1\|,\vspace{-2mm}
\]

where $s = -\ln\|\xi\|_{\dn}$, $C_1=\sup_{0.5\leq \|y\|\leq 1.5} \|B\frac{\partial u_{hom}(y)}{\partial y}\|$.
Since $C_2=\sup_{\|z\|=1} \|\frac{\partial V}{\partial z}\|<+\infty$ then \vspace{-2mm}
\[
\tfrac{\partial V}{\partial x} f(\xi,\delta) \leq \tfrac{\partial V}{\partial x} f(\xi,0)+ V^{1+\mu} \tilde C\sum_{i=1}^2\|\dn(-\ln \|\xi\|_{\dn})\delta_i\|,\vspace{-2mm}
\]
where $\tilde C=C_2\max\{C_1,1\}$. In Theorem \ref{thm:hybrid_control} it is shown that 
$\tfrac{\partial V}{\partial x} f(\xi,0)\le -\tilde{\rho}V$ for all $\xi \in \mathcal{C}.$ 
If $\xi\in\mathcal{C}$ and  $\|\dn(-\ln \|\xi\|_{\dn})\delta_i\|\leq\frac{\tilde \rho}{4\tilde C}, i=1,2$
imply
$\tfrac{\partial V}{\partial x} f(\xi,\delta)\leq -\frac{\tilde \rho}{2} V^{1+\mu}.$
Since $\|\dn(-\ln \|\xi\|_{\dn})\delta_i\|\leq\frac{\tilde \rho}{4\tilde C}$ is equivalent to 
$\|4\tilde C\delta_i/\tilde \rho\|\leq \|\xi\|_{\dn}$ then for $\xi\in \mathcal{C}$ we have \vspace{-2mm}
\begin{equation}\label{eq:ISS_LF}
 V(\xi)\geq \|4\tilde C\delta_i/\tilde \rho\| \quad  \Rightarrow \quad \dot V(\xi)\leq -\tfrac{\tilde\rho}{2}V^{1+\mu}(\xi).\vspace{-2mm}
\end{equation}
In the case of the impulse-free continuous time system, the latter implication guarantees  that $V$ is an  ISS Lyapunov function (see, \cite{SontagWang1996:SCL}). Since $\mathcal{C}$ is open connected set and $0\in \mathcal{C}$ then for sufficiently small perturbations and initial conditions the system \eqref{eq:att_err_pert} has no impulse so it is locally ISS. Using the continuity of $f$ and smoothness of $V$ we expand \eqref{eq:ISS_LF}
to $\xi\in \overline{\mathcal{C}}\backslash\mathcal{D}$. In Theorem \ref{thm:hybrid_control} it is shown that $V(\Pi \xi)\leq V(\xi)$ for any $\xi\in \mathcal{D}$. Hence, if the impulsive system \eqref{eq:att_err_pert} has solutions (in the sense of  \cite{Hespanha_etal2008:Aut}) belonging to $\mathcal{C}$ almost everywhere then the function $t\mapsto V(\xi(t))$ with $t\geq t_0$ is strictly decreasing  as long as $ V(\xi(t))\geq \frac{4\tilde C\|\delta_i\|_{L^{\infty}(t_0,t)}}{\tilde \rho}$. Using conventional arguments (see, \cite{SontagWang1996:SCL}, \cite{Hespanha_etal2008:Aut}) we derive  ISS of the impulsive system. Moreover, from 
\eqref{eq:ISS_LF} with $\mu<0$ it follows the finite-time ISS (\cite{Hong_etal2010:SIAM}), but the fixed-time ISS (\cite{Lopez-Ramirez_etal2020:SCL}, \cite{Aleksandrov_etal2022:MCSS}) of the system \eqref{eq:att_err_pert}  with $u=u_{fxt}$ follows from \eqref{eq:ISS_LF} with $\mu_2<0$ for $V(\xi)\leq 1$ and from \eqref{eq:ISS_LF} with $\mu=\mu_1>0$ for $V(\xi)>1$. The proof is complete.
\end{proof}

\section{Numerical Simulation}
\label{sec:simu}

 A simulation of a quadrotor attitude tracking is created  in Simulink with the  step size $0.001$ sec. The model corresponds to  QDrone platform manufactured by Quanser: $J= 10^{-2}\text{diag}(1.0, 0.82 ,1.48)$ kg$\cdot$ m$^2$.
The initial states are given by   {$\omega(0)\!=\![0,-1,0]^{\top}$rad$\cdot$ sec$^{-1}$ and   $R(0)\!=\!I$.

The parameters of the homogeneous controller \eqref{eq:hom_con_Q} are selected as follows $\mu=-1/3$ and  $K = [-12 I_3, -6I_3]$, $G_{\dn}=
\left[\begin{smallmatrix}
	\tfrac{4}{3}I_3 & \boldsymbol{0}\\
	\boldsymbol{0} & I_3
\end{smallmatrix}\right]$, 
 $\varepsilon=0.05$, $P=\left[\begin{smallmatrix}
I_3 & 0.05I_3\\
0.05I_3 & 0.083I_3
\end{smallmatrix}\right]$.}
{ Hence, we have $\tilde{\rho}=0.3865$.
	
	For comparison, we consider the finite-time controller proposed in \cite{shi2018global}, \cite{shi2017almost}\vspace{-2mm}
\begin{equation}
	u_1 = -J_r^{-\top}k_p  sign(\theta_e)^{\alpha_1}  - k_d  sign(\omega_e)^{\alpha_2}\vspace{-2mm}
\end{equation}
where $\alpha_1 = \alpha$, $\alpha_2=2\alpha/(1+\alpha)$, $\alpha\in(0,1)$; $x=[x_1,x_2,x_3]^\top\in\mathbb{R}^3$, $ sign(x)^\alpha=[|x_2|^\alpha sign(x_1),$ $
|x_2|^\alpha sign(x_2),|x_3|^\alpha sign(x_3)]^\top$.
Since the tracking error system with $u_1$ is $\dn$-homogeneous  with degree $\tfrac{\alpha-1}{\alpha +1}$, so, for a fair comparison, we select $\alpha = 0.5\Rightarrow \tfrac{\alpha-1}{\alpha +1}=\mu =-1/3$ to have the same homogeneity degree for both controller. 

 The linear controller proposed in \cite{yu2016global}: \vspace{-2mm}
\begin{equation}
	u_2 = -k^{l}_p \theta_e  - k^{l}_d \omega_e \vspace{-2mm}
\end{equation}
is also utilized for comparison.  The controller $u_2$ corresponds to our homogeneous controller if $\mu=0$. 
We use the same gains for all controllers: $K_1 = k_p = k_p^l = 12I_3$ and $K_2 = k_d = k_d^l = 6I_3$ (the gain selection takes into consideration that each element of torque is bounded by $[-0.5, 0.5]$).  In all cases, the torque $M$ is given by \eqref{eq:att_control}.

The desired trajectory is generated by 
 $\omega_d \!=\! [\text{--}0.2t,\text{--}0.2t+3,t]^\top$ and $R_d$ obtained by integration of $\dot R_d=R_d w_d^{\wedge}$ with the initial condition $R_d(0) = e^{([-0.1\pi, 0.95\pi, 0]^\top)^\wedge}$.  
 Two cases are studied by numerical  simulations:
\\
 \textit{Case 1}: no measurement noises or perturbations;\\ 
 \textit{Case 2}: the measurement noise $\delta_1$ is modeled by the uniformly distributed random variables $d_{R}$ and $d_{\omega}$ satisfy $\|d_R\|_{\infty}\le 0.05$ and $\|d_{\omega}\|_{\infty}\le 0.05$, but the additive perturbation is deterministic, state-dependent and defined as $\delta_2 = [\boldsymbol{0}, 0.01\operatorname{diag}(2,2,1)\omega)]^\top$.\\
	 The simulation results are shown in Fig. \ref{fig:1}, \ref{fig:2}, and \ref{fig:4}. The control energy $E = (\smallint_{0}^{t_f}\|M\|^2d\tau)^{1/2}$ on the time interval  $[0,t_f]$  is compared in Table 1. In the disturbance-free case, as shown in Fig. \ref{fig:1}, the system with finite-time controls exhibits faster convergence compared to the linear one. The proposed homogeneous control $u_{hom}$ is faster than $u_1$ and requires less control energy.
	 Figure \ref{fig:2} illustrates the state dependent jump and the decay of the canonical homogeneous norm during the jump.
	 Measurement noises and perturbations destroy   the ideal asymptotic and finite-time stability properties. However, in this specific case, the finite-time controllers $u_{hom}$ and $u_1$ ensure higher accuracy than the linear controller (see Fig. \ref{fig:4}). Similarly to the disturbance free case the homogeneous controller $u_{hom}$ has faster response that the finite-time controller $u_1$. Additionally, the proposed homogeneous controller $u_{hom}$ consuming less energy providing the better control quality (faster response and/or better precision). 
	 
	For the control homogeneous controller $u_{hom}$, the settling time estimate  obtained in Theorem 3 gives $10.99$ sec . This estimate is rather conservative, but it is obtained for the first time for the homogeneity-based global attitude  control system.
}
\begin{figure}[thpb]
	\centering
	\includegraphics[width=0.40\textwidth]{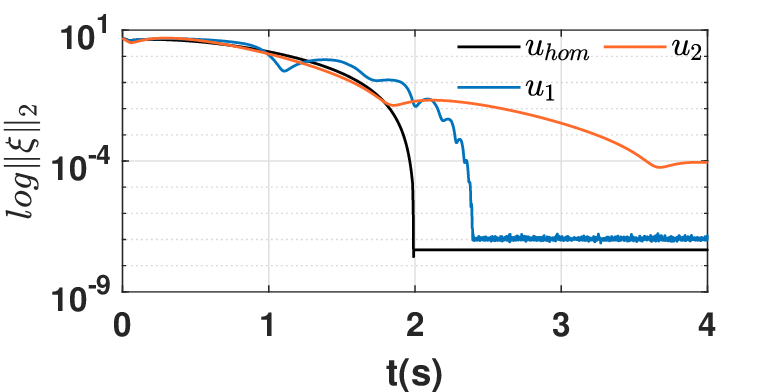}
	\caption{The  norm of the tracking error $\xi$ (Case 1).}
	\label{fig:1}
\end{figure}
\vspace{-2mm}
\begin{figure}[thpb]
	\centering
	\includegraphics[width=0.33\textwidth]{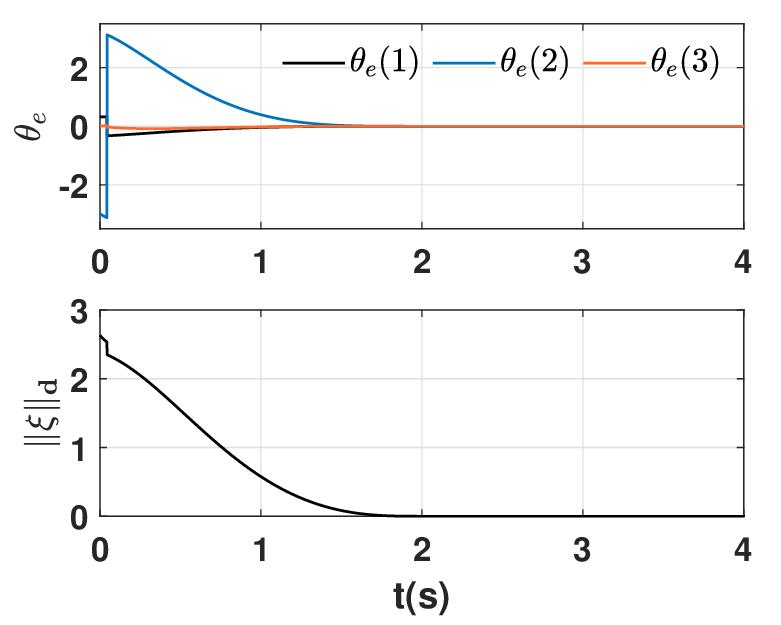}
	\caption{The jump behavior under proposed homogeneous controller (Case 1).}
	\label{fig:2}
\end{figure}
\vspace{-2mm}

\begin{figure}[thpb]
	\centering
	\includegraphics[width=0.33\textwidth]{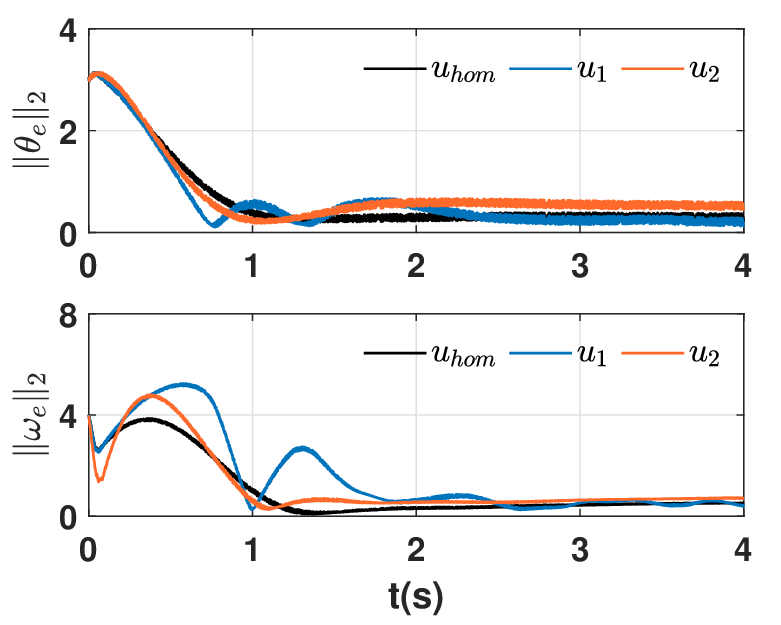}
	\caption{The norm of the tracking error (Case 2).}
	\label{fig:4}
\end{figure}

\begin{table}[h]
	\caption{The control energy $E$ for $t_f=4$.}
		\centering
		\begin{tabular}{ |p{0.5cm}|p{0.6cm}|p{0.6cm}|p{0.6cm}|p{0.6cm}|p{0.6cm}|p{0.6cm}|}		
		\hline
		& \multicolumn{3}{|c|}{Case 1} & \multicolumn{3}{|c|}{Case 2} \\
		\cline{2-7}
		& $u_{hom}$ & $u_1$ & $u_2$ & $u_{hom}$ & $u_1$ & $u_2$\\
		\hline
		$E$ & 0.104 & 0.134 & 0.153 & 0.146 & 0.211 & 0.190\\
		\hline
	\end{tabular}

\label{table:1}
\end{table}

\section{Conclusion}
In this paper, a global attitude tracking control for a fully-actuated rigid body is developed using various symmetries of the system. First, a rotation symmetry of the system is utilized in order to derive an impulsive linear system, which describes the rigid-body dynamics on the Lie algebra $\mathfrak{so}(3)$.  Next, a dilation symmetry (generalized homogeneity) is applied for finite/fixed-time controller design. The control law is formulated in terms of the canonical homogeneous norm being the {\Yu strict}  Lyapunov function of the closed-loop error system. For homogeneity-based methods, a global estimate of the settling time is provided for the first time. This estimate is exact, at least, for small initial errors.
Theoretical results are supported by numerical simulations.

\begin{ack}                               
This work is supported by CSC (China Scholarship Council) Grant 202006030019 and the National Natural Science
Foundation of China under Grant 62050410352.
\end{ack}

\section*{Appendix: $SO(3)$ group and $\mathfrak{so}(3)$ algebra}

Let us recall  briefly some results  (see, e.g.,  \cite{chirikjian2012information} for more details) about $SO(3)$  and $\mathfrak{so}(3)$.\vspace{-2mm}
\begin{equation}
	\begin{aligned}
		SO(3)&=\{R \in \mathbb{R}^{3 \times 3}: R R^{T}=I_3, \ \det(R)=1\},\\
		\mathfrak{so}(3)&=\{X\in\mathbb{R}^{3\times3} : X^{\top}=-X\}.
	\end{aligned}\vspace{-2mm}
\end{equation}
The  $\mathfrak{so}(3)$ consists of $3\times 3$ skew-matrices with the Lie bracket given by the commutator $[X_1,X_2]=X_1X_2-X_2X_1$, where $X_1,X_2\in 	\mathfrak{so}(3)$.
Since  any $X\in \mathfrak{so}(3)$ can be uniquely represented by $X=x^{\wedge}$ with $x\in\mathbb{R}^3$, then  $\mathfrak{so}(3)$ is isomorphic to $\R^3$.
 Notice also that for $x\in\R^3$ and $R\in SO(3)$ one holds $Rx^{\wedge}R^{\top}=(Rx)^{\wedge}$.
 The exponential map $\mathrm{exp}(\cdot)$ is a surjective map from $\mathfrak{so}(3)$ to  $SO(3)$. For
 $X\in \mathfrak{so}(3)$
the exponential map is given by 
$	\exp \left(X\right)=\sum_{n=0}^{\infty} \tfrac{1}{n !}\left(X\right)^{n}\in SO(3)
$.
Any  $X\in \mathfrak{so}(3)$ (resp. any vector $x\in \R^3 : X=x^{\wedge}$) uniquely defines a  $R=\exp(X)\in SO(3)$.
For $R\in SO(3) : \mathrm{trace}(R)\neq -1$, the exponential map  $\mathrm{exp}(\cdot)$ and its inverse map $\log(\cdot)$ can be defined as follows:\vspace{-2mm}
\begin{equation}\label{eq:exp_log}
		\begin{aligned}
			R=\exp(X)=&I_3\!+\!\tfrac{\sin|x|}{|x|}X\!+\! \tfrac{1-\cos|x|}{|x|^2}X^2 \\
			X=\log(R)=&\tfrac{\phi}{2\sin\phi}\left(R-R^{\top}\right) 
		\end{aligned}\vspace{-2mm}
\end{equation}
where  $\phi\in (-\pi, \pi)$ is a solution of the equation $1+2\cos\phi=\mathrm{trace}(R)$. 
For $\mathrm{trace}(R)=-1$ then there is an ambiguity in the definition of $X$, since in this case $X=(\pm \pi y){\tiny }^\wedge$, where  $y\in\mathbb{R}^3$ is a unique unit vector  ($y^{\top}y=1$) satisfying $Ry=y$.  
In fact,  $SO(3)$ group in the exponential coordinates has a correspondence with the closed ball $B^3[\pi]$. The  transformation $\mathrm{Log}(\cdot)=\log (\cdot^{\wedge}):\R^3\to SO(3)$ is a bijection on $B^3(\pi)$.  The ambiguity appear on the sphere $S^3(\pi)$. 
The left Jacobian of the group $SO(3)$ is\vspace{-2mm} \begin{equation}\label{eq:J_r}
	\begin{aligned}
		J_{\textit{r}}(x^\wedge) & =I-\tfrac{1-\cos |x|}{|x|^2}x^\wedge+\tfrac{|x|-\sin |x|}{|x|^3}(x^\wedge)^2, \\
		J_{\textit{r}}^{-1}(x^\wedge) & =I+\tfrac{1}{2}x^\wedge+\left(\tfrac{1}{|x|^2}-\tfrac{1+\cos |x|}{2|x| \sin |x|}\right)(x^\wedge)^2.
	\end{aligned}\vspace{-2mm}
\end{equation}
 Hence, $J_{\textit{r}}^{-1}$ bounded for $|\theta_e|\in[0,\pi]$.

\bibliographystyle{unsrtnat}        
\bibliography{bib_all} 

\end{document}